\documentclass[letterpaper, 10 pt, conference]{ieeeconf}  
                                                       
\IEEEoverridecommandlockouts                              
\usepackage{amsmath} 
\usepackage[nolist]{acronym}
\usepackage{xcolor}
\usepackage{amsmath}
\usepackage{amssymb}
\usepackage{mathtools}
\usepackage{subcaption}
\usepackage{hyperref}
\hypersetup{colorlinks=true, linkcolor=black}
\usepackage[ruled]{algorithm2e}
\usepackage{mathrsfs}
\usepackage{circuitikz}
\definecolor{ccolor}{RGB}{203,96,21}

\newcommand{\urg}[1]{{\color{red} #1}}

\newcommand{\ct}[1]{\urg{[cite]}}
 
\newcommand{\R}{\mathbb{R}}

\newcommand{\N}{\mathbb{N}}

\newcommand{\Lie}{\mathcal{L}} 
 
\newcommand{\Z}{\mathbb{Z}}

\newcommand{\vs}{\mathcal{V}}
\newcommand{\es}{\mathcal{E}}
\newcommand{\gs}{\mathcal{G}}

\DeclarePairedDelimiter{\abs}{\lvert}{\rvert}
\DeclarePairedDelimiter{\norm}{\lVert}{\rVert}

\DeclarePairedDelimiterX{\inp}[2]{\langle}{\rangle}{#1, #2}

\DeclarePairedDelimiter{\Mp}{\mathcal{M}_+(}{)}

\DeclareMathOperator*{\argmax}{arg\!\,max}

\newtheorem{thm}{Theorem}[section]

\newtheorem{assum}[thm]{Assumption}
\newtheorem{prop}[thm]{Proposition}
\newtheorem{prob}[thm]{Problem}

\newtheorem{rmk}{Remark} 

\usepackage{color}
\usepackage{ulem}

\title{\LARGE \bf 
Optimal Pulse Patterns through a Hybrid Optimal Control Perspective
}

\author{Jared Miller and Petros Karamanakos
\thanks{J.~Miller is with the Chair of Mathematical Systems Theory, Department of Mathematics,  University of Stuttgart, Stuttgart, Germany; e-mail: jared.miller@imng.uni-stuttgart.de.}
\thanks{P.~Karamanakos is with the Faculty of Information Technology and Communication Sciences, Tampere University, 33101 Tampere, Finland; e-mail: p.karamanakos@ieee.org.}
}

\begin{document}

\maketitle
\thispagestyle{empty}
\pagestyle{empty}


\begin{abstract}
\label{sec:abstract} Optimal pulse patterns (OPPs) are a modulation method in which the switching angles and levels of a switching signal are computed via an offline optimization procedure to minimize a performance metric, typically the harmonic distortions of the load current. Additional constraints can be incorporated into the optimization problem to achieve secondary objectives, such as the limitation of specific harmonics or the reduction of power converter losses.
The resulting optimization problem, however, is highly nonconvex, featuring a trigonometric objective function and constraints as well as both real- and integer-valued optimization variables. This work casts the task of OPP synthesis for a multilevel converter as an optimal control problem of a hybrid system. This problem is in turn lifted into a convex but infinite-dimensional conic program of occupation measures using established methods in convex relaxations of optimal control. Lower bounds on the minimum achievable harmonic distortion are acquired by solving a sequence of semidefinite programs via the moment-sum-of-squares hierarchy, where each semidefinite program scales in a jointly linear manner with the numbers of permitted switching transitions and converter voltage levels.
\end{abstract}

\section{Introduction}
\label{sec:introduction}

Power converters are a key technology for industry and the green-energy transition~\cite{zhong2012control, orlowska2016industrial}. 
However, the switching nature of the converters makes their control a challenging task~\cite{kouro2010recent}.
Although modulation helps mask this characteristic, the performance of the power electronic system deteriorates at low switching-to-fundamental frequency ratios.
Under such conditions, optimizing the switching signals can deliver substantial gains, including lower harmonic distortion, higher efficiency, and compliance with harmonic grid standards.

In this context, programmed modulation methods, such as selective harmonic elimination (SHE) and optimal pulse patterns (OPPs), are promising candidates.
SHE eliminates targeted harmonics by solving a nonlinear system of equations to determine the switching angles, i.e., the switching time instants, of the switching signal~\cite{turnbull1964selected, patel1973generalized, dahidah2014review}.
However, this approach typically yields multiple solutions with different characteristics, requiring a post-processing step to select the most suitable one~\cite{chiasson2004unified}.
Moreover, SHE may generate infinitely many solutions in cases of degeneracy or underconstrained setups, while it may yield no feasible solution as the number of unknowns increases~\cite{wells2005selective}.



OPPs, on the other hand, are obtained through an optimization procedure, where the objective captures a performance metric such as the output current harmonic distortion~\cite{buja1980optimum}. The optimization yields the switching angles and corresponding switch positions (i.e., converter output levels) of the switching signal. Hence, the OPP problem involves both continuous and discrete decision variables. Together with the trigonometric form of the objective function and constraints, this results in a mixed-integer, nonconvex problem that is challenging to solve. For this reason, OPPs are computed offline and stored in look-up tables for real-time use.

The optimization-based formulation, however, provides design flexibility. OPPs can be tailored to relax symmetry and switching constraints~\cite{birth2019generalized}, mitigate specific harmonics~\cite{rahmanpour2023three}, reduce common-mode voltage~\cite{koukoula2024optimal}, or limit converter losses~\cite{geyer2023optimized}. A wide range of optimization approaches has been investigated, including iterative solvers~\cite{meili2006optimized}, heuristic methods such as genetic and swarm algorithms~\cite{dahidah2008hybrid, kavouski2012bee}, virtual-angle formulations~\cite{koukoula2024fast}, gradient-based optimization~\cite{ali2024optimal}, as well as learning-based methods, such as deep reinforcement learning and differentiable programming~\cite{qashqai2020new, abu2025optimized, abu2025diff}. These methods typically provide locally optimal solutions, but cannot  guarantee global optimality of the proposed solution.


Given the nonconvexity of the OPP problem, a critical task is to establish lower bounds on the harmonic distortion, which can indicate if a feasible pulse pattern is globally optimal. The most closely related work to this paper is~\cite{wachter2021convex}, which applies sum-of-squares (SOS) polynomial optimization techniques~\cite{lasserre2009moments} to lower-bound differential-mode harmonic distortion under the requirement of a fixed sequence of switching transitions. A major obstacle in applying polynomial optimization to OPPs lies in the trigonometric terms appearing in the harmonic constraints and objective. 
To address this,~\cite{wachter2021convex} approximates the constraints via Taylor polynomials with bounding residuals, and the objective via low-order interpolating polynomials.


This paper adopts a time-domain perspective, casting the OPP synthesis problem as an optimal control problem (OCP) for a mode-scheduled hybrid system \cite{teel2012hybrid, hale2014mode}. This hybrid control problem is lifted into an infinite dimensional linear program (LP) in measures~\cite{claeys2016modal, zhao2017optimal, miller2023peakhy}, using standard methods from the literature of optimal control~\cite{rubio1975generalized}. The proposed OCP-relaxed infinite dimensional LP in occupation measures has solely polynomial-valued data, and can thus be discretized through SOS methods without requiring Taylor approximation or \textit{a-priori} fixing of the switching sequence. The size of the semidefinite program increases linearly with both the number of converter levels and the number of permitted switching transitions, while the computational complexity grows jointly sub-quadratically.

The contributions of this work are:
\begin{itemize}
    \item A framework to analyze OPP synthesis through the lens of hybrid system optimal control.
    \item The introduction of an infinite-dimensional LP that provides a lower bound on the minimal current total demand distortion (TDD) under imposed constraints.
    \item Truncation and recovery by the moment-SOS hierarchy.
    \item Numerical results based on a five-level converter.
\end{itemize}

To the best of the authors' knowledge, optimal control theory and occupation measures have not yet been applied to the design of OPPs.
\section{Preliminaries}
\label{sec:preliminaries}
\begin{acronym}



\acro{FW}{Full-Wave}

\acro{HW}{Half-Wave}
\acro{LMI}{Linear Matrix Inequality}
\acroplural{LMI}[LMIs]{Linear Matrix Inequalities}
\acroindefinite{LMI}{an}{a}


\acro{LP}{Linear Program}
\acroindefinite{LP}{an}{a}
\acro{MPC}{Model Predictive Control}

\acro{OCP}{Optimal Control Problem}
\acroindefinite{OCP}{an}{a}

\acro{OPP}{Optimal Pulse Pattern}
\acroindefinite{OPP}{an}{a}



\acro{PSD}{Positive Semidefinite}

\acro{QW}[QaHW]{Quarter-and-Half-Wave}



\acro{SDP}{Semidefinite Program}
\acroindefinite{SDP}{an}{a}

\acro{SHE}{Selective Harmonics Elimination}
\acroindefinite{SHE}{an}{a}

\acro{SOS}{Sum of Squares}
\acroindefinite{SOS}{an}{a}

\acro{THD}{Total Harmonic Distortion}
\acro{TDD}{Total Demand Distortion}


\end{acronym}

\subsection{Notation}
The $\ell$-dimensional vector space of real numbers is $\R^\ell$. The set of integers is $\Z$, the subset of nonnegative integers is $\Z_{\geq 0}$, and the further subset of natural numbers is $\N$. The set of integers between $a$ and $b$ (inclusive) is $a..b$. Given vectors $c_1, c_2 \in \R^n$, the comparator $\leq$ ($\geq$) will hold elementwise as $c_1 \leq c_2$ ($c_1 \geq c_2$). 
The unit ball in 2 dimensions is $B = \{(c, s) \in \R^2 \mid c^2+s^2 = 1\}$. The map $\psi: [0, 2\pi] \mapsto B$ implements a trigonometric lift $\psi(\theta) = (\cos(\theta), \sin(\theta))$.
Given angles $\alpha_1, \alpha_2 \in [0, 2\pi]$ with $\alpha_1 \leq \alpha_2$, the symbol $B([\alpha_1, \alpha_2])$ denotes the arc of the unit circle between angles $\alpha_1 \leq \alpha_2$ (image of $[\alpha_1, \alpha_2]$ along $\psi$).

\subsection{Optimal Pulse Patterns}

The goal of OPPs is to minimize harmonic distortions subject to switching, device, and harmonic constraints. Table~\ref{tab:constraints} summarizes the constraints involved in an OPP problem, with a \checkmark placed next to designer-chosen requirements in this work. Moreover, Table~\ref{tab:opp_tdd_param} summarizes the parameters used to describe an OPP problem with the constraints in Table~\ref{tab:constraints}.

\begin{table}[t!]
    \centering
        \caption{Constraints in OPP}
    \begin{tabular}{cl}         
         Design & Constraint Description \\ \hline                     
           \checkmark & Number of switches \\           
           \checkmark & Unipolarity \\
          & Harmonics specifications  \\
           & Interlocking angle \\
           & Converter levels \\
           & Cannot jump by $> 1$ level\\
           & $2\pi$-periodicity \\           
    \end{tabular}
    \label{tab:constraints}
\end{table}

\begin{table}[t!]
\caption{\label{tab:opp_tdd_param} Parameters for OPP problem description}
\centering
\begin{tabular}{cl}
\textbf{Parameter}              & \textbf{Description}                      \\ \hline
$L$                    & Converter levels                \\
$d$                    & Pulse number $(k = 4d)$         \\
$\Theta$               & Interlocking angle         \\
$q, h$        &  Harmonics  constraints 
\end{tabular}
\vspace{-8pt}
\end{table}

The possible switch positions of an $N$-level converter, corresponding to its output voltage levels, are contained in a vector $L = \{v_{n}\}_{n=1}^N$. OPPs involve the creation of a $2\pi$-periodic switching signal $u: [0, 2\pi] \rightarrow L$. Assuming that the dc source has a constant value $V_{dc}$ and the converter chooses a switch position $u(\theta) \in L$, the converter output voltage is $V_{\text{out}}(\theta) = (V_{dc}/2) u(\theta)$. 

Assuming an inductive load $L_\text{load}$, the load current dynamics are $L_\text{load} \dot{I}_\text{load}(\theta) = V_{\text{out}}= (V_{dc}/2) u(\theta)$. In the sequel of this work, normalized (per unit) values are considered as per $I = (2L_\text{load}/V_{dc}) I_\text{load}$, allowing for single-integrator dynamics as $\dot{I} = u$.



A $k$-switching signal $u(\theta)$ can be parameterized by switching angles $\{\alpha^i\}_{i=1}^{k}$ and switch positions $\{u^i\}_{i=0}^k$ with
\begin{align}
\label{eq:v_theta_step}
    u(\theta) &= \begin{cases}  u^0 & \theta \in [0, \alpha^1) \\      
        u^i & \theta \in [\alpha^{i}, \alpha^{i+1}), \ i \in 1..k-1\\        
        u^k & \theta \in [\alpha^k, 2\pi)
    \end{cases}.
\end{align}

The number of switch positions in a signal can be parameterized in terms of its half-integer \textit{pulse number} $d = k/4$. Symmetries can be imposed on the switching signal to reduce the computational complexity of the SHE- and OPP-based procedures. These symmetries include
\begin{subequations}
\label{eq:symmetries}
\begin{align}
& \text{Full-Wave} &  u(\theta+2\pi) &=u(\theta) & & \forall \theta \in \R \\
    & \text{Half-Wave} &  u(\theta+\pi) &= -u(\theta) & & \forall \theta \in [0, \pi] \\    
    & \text{Quarter-Wave} &  u(\theta) &= u(\pi - \theta) & & \forall \theta \in [0, \pi].
\end{align}
\end{subequations}
Quarter-and-half-wave (QaHW) symmetry implies that the signal $u(\theta)$ is zero-mean. If the induced current $I(\theta)$ driven by $u(\theta)$ is likewise zero-mean, then $\forall \theta \in [0, \pi]: I(\pi-\theta) = -I(\theta)$. Figure \ref{fig:symmetries} visualizes the voltage $u(\theta)$ (top) and the current $I(\theta)$ (bottom) of a $d=3, \ N=3$ QaHW signal. 

\begin{figure}[t!]
    \centering
    \includegraphics[width=0.7\linewidth]{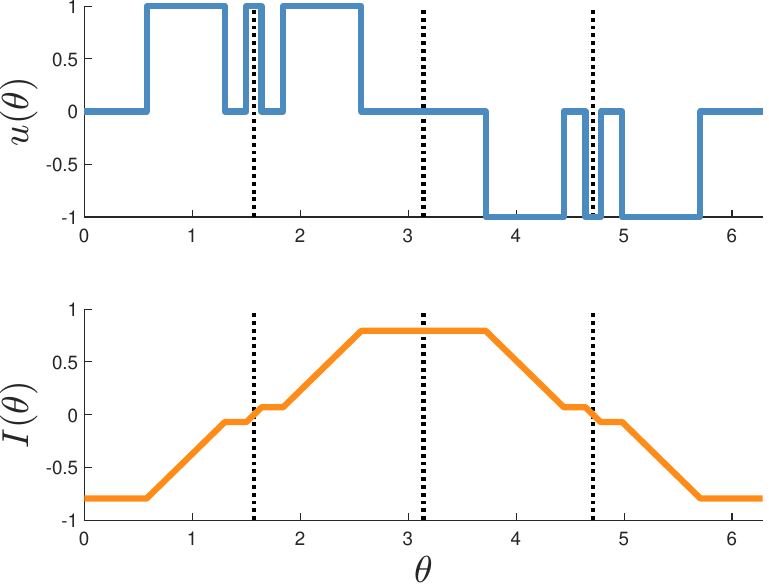}
    \caption{ QaHW signals $u(\theta)$ and $I(\theta)$ with $L = \{-1, 0, 1\}$}
    \label{fig:symmetries}
    \vspace{-6pt}
\end{figure}

\ac{QW} signals always have integer pulse numbers $d$ in the first quarter of the fundamental period.
The switching angles of a \ac{QW} signal over the full angle range $[0, 2\pi]$ are entirely defined by the first $d$ angles as:
\begin{align}
    \{\alpha^i\}_{i=1}^{4d} = \begin{Bmatrix}
            \{\alpha^i\}_{i=1}^{d}, & \{\pi - \alpha^{d-i+1}\}_{i=1}^{d},  \\
            \{\pi + \alpha^i\}_{i=1}^{d}, & \{2\pi - \alpha^{d-i+1}\}_{i=1}^{d}
    \end{Bmatrix}.
\end{align}




The variables of the OPP problem will be the switching angles $\{\alpha^i\}_{i=1}^{d}$ and the indices $\{n^i\}_{i=0}^d$ such that the output levels satisfy $u^i = u_{n^i}$. A constraint $n^i - n^{i-1} \in \{-1, 1\}$ is imposed to require that the switching signal steps up or down by only one level at any switching angle.

The $N$-level converter operates at a fundamental frequency of $f_1$ (e.g., 50 Hz) with an angular frequency of $\omega_1 = 2 \pi f_1$.   Consecutive switching transitions must be separated by at least an \textit{interlocking time} of $T_s$ time (e.g., 40 $\mu$s) to prevent short circuits. The interlocking time produces an interlocking angle $\Theta = \omega_1 T_s$ between any two consecutive switching transitions over the $2\pi$-periodic operation.
The \ac{QW}-symmetric signal has a Fourier series  $u(\theta) = \sum_{\ell=1}^\infty b_\ell \sin(\ell \theta)$ with
\begin{align}
     &     b_\ell = \begin{cases}
        0 & 2\ell \in \Z \\
        \frac{4}{\ell \pi} \sum_{i=1}^d \left(u^i - u^{i-1}\right) \cos(\ell \alpha^i) & \text{else}. \label{eq:fourier_coeff}
    \end{cases}
\end{align}


If the applied signal $u(\theta)$ to the inductor is zero-mean and the $I(\theta)$ load current is likewise zero mean, then the signal energy of the current  $I(\theta)$ satisfies
\begin{align}
    \norm{I}_2^2 &= \textstyle \int_{\theta=0}^{2\pi} I(\theta)^2 d \theta \nonumber =  4\int_{\theta=0}^{\pi/2} I(\theta)^2 d \theta.
    \intertext{The signal energy can be expressed as a sum over Fourier coefficients via Parseval's theorem for real Fourier series as}
     \norm{I}_2^2&= \textstyle \pi\left(\sum_{\ell\geq 2}^\infty\frac{ b^2_\ell}{(2 \pi f_1 \ell)^2}\right) = \frac{1}{4 \pi f_1^2 } \left(\sum_{\ell\geq 2}^\infty\frac{ b^2_\ell}{\ell^2}\right). \label{eq:current_parseval}
\end{align}
The term $\norm{I}_2^2$ has a closed-form expression as a cubic polynomial in the switching angles $\alpha$ and the initial current $I(0)$ for fixed $u$ (under the angle-extended notation $\alpha^0=0, \ \alpha^{d+1} = \pi/2)$:
\begin{align}
    I^0 &= I(0) & I^i &= I^{i-1} + u^{i-1}(\alpha^i - \alpha^{i-1}),  \quad \forall i \in 1..d \nonumber    \\
    E^0 &= 0 & E^i &= \begin{cases}
        (I^{i-1})^2 (\alpha^i - \alpha^{i-1}) & u^i = 0 \\
        ((I^i)^3 - (I^{i-1})^3)/(3 u^i) & u^i \neq 0
    \end{cases} \nonumber    
\end{align}
\begin{align}
        \norm{I}_2^2 &= 4 \textstyle \sum_{i=0}^{d} E^i. 
\end{align}
%
The current TDD is defined w.r.t.~the rated rms value $I_R$:
\begin{align}
    \text{TDD}_I &= 
    \frac{1}{\sqrt{2} I_R \omega_1 L_{\text{load}}  } \frac{V_{dc}}{2} \sqrt{\norm{I}_2^2/\pi -  b_1^2}. \label{eq:tdd_current}
\end{align}

The current TDD is a measure of spectral efficiency: low values of it indicate that the majority of signal energy is present inside the fundamental harmonic.
Voltage harmonics constraints for the switching sequence will be specified by an  odd-integer vector $q$ and a convex set $h$. 
As an example, the constraint $b_1 = M$, $b_3 = 0$,  $b_5 \in [-0.02, 0.02]$, $b_7 \in [0, 0.1]$ can be encoded as 
\begin{align}
    q &= [1; 3; 5; 7], & h = \{M\} \times \{0\} \times [-0.02, 0.02] \times [0, 0.1].\nonumber
\end{align}


The minimal-energy OPP synthesis problem is as follows:
\begin{prob}
    Given parameters in Table \ref{tab:opp_tdd_param}, find a switching sequence ($\{\alpha^i\}, \{n^i\}$) with $u^i = v_{n^i}$ to minimize:
    \label{prob:tdd_orig}
\begin{subequations}
    \label{eq:tdd_orig_alpha}
\begin{align}
    J^* = &\min_{\alpha^i, n^i} \norm{I}_2^2 \label{eq:tdd_orig_obj}\\
   \text{s.t.}\quad & \textstyle  \frac{4}{\pi} \int_{\theta=0}^{\pi/2} \sin(q \theta) u(\theta) d \theta \in h 
    \label{eq:tdd_orig_alpha_she} \\
    & n^{i+1} - n^i \in \{-1, 1\} & & \forall i \in 0..d-1 \label{eq:tdd_orig_step}\\
    & \alpha^{i+1} \geq \alpha^i + \Theta &  &\forall i \in 1..d-1\label{eq:tdd_orig_alpha_lock}\\ 
    & \alpha^1 \geq \Theta/2, \ \pi/2 \geq \alpha^d + \Theta/2 & \label{eq:tdd_orig_alpha_lock_end}\\    
    & n^i \in 1..N & &  \forall i \in 0..d \label{eq:tdd_orig_alpha_L}\\
    & n^0 = n^d. \label{eq:tdd_orig_alpha_n_end}
\end{align}
\end{subequations}
\end{prob}
The harmonics constraints in \eqref{eq:tdd_orig_alpha_she} exploit the \ac{QW} symmetry of the signal to only integrate between $\theta \in [0, \pi/2]$.
The formalism in \eqref{eq:tdd_orig_alpha} allows for inhomogenous spacing between levels \cite{dahidah2014review}.
In the case where the fundamental mode $b_1$ is fixed to a constant values (e.g., $b_1 = M$), synthesis of a signal $u$ that minimizes $\norm{I}^2_2$ will in turn minimize $\text{TDD}_I$. Problem \ref{prob:tdd_orig} involves integer-valued variables \eqref{eq:tdd_orig_alpha_L}, transcendental constraints \eqref{eq:tdd_orig_alpha_she}, and nonlinearities in the objective \eqref{eq:tdd_orig_obj}.

\section{OPPs as Hybrid Control}
\label{sec:opp_as_ocp}

The OPP task in Problem \ref{prob:tdd_orig} is a nonconvex optimization problem over the finite-dimensional variables $(\{\alpha^i\}_{i=0}^{d+1}, \{n^i\}_{i=0}^{d})$. This section explores how Problem \ref{prob:tdd_orig} can be cast as an OCP of a hybrid dynamical system. 

The motivation for this conversion is in treating the  simultaneous presence of the angle $\alpha$ in $\norm{u}_2^2$ and trigonometric expressions of $\alpha$ in the Fourier coefficients \eqref{eq:fourier_coeff}. If only trigonometric terms were present, then a trigonometric transformation $(\cos(\alpha^i), \sin(\alpha^i)) \rightarrow (c_i, s_i)$ can be used to represent the harmonics constraints as polynomials (as performed in SHE). The fundamental insight undergirding this section is that the angle $\alpha^i$ corresponds to the time required to traverse the unit circle at $1\,$rad/s when starting at $(1, 0)$ and ending at $(\cos(\alpha^i), \sin(\alpha^i))$.


\subsection{Assumptions}
The following assumptions are required given Table \ref{tab:opp_tdd_param}:
\begin{assum}
    The pulse number $d$ is an integer.
    \label{assum:switch}
\end{assum}
\begin{assum}
    The interlocking angle $\Theta$ is positive.
    \label{assum:interlock}
\end{assum}
\begin{assum}
$L$ is sorted in increasing order, is symmetric ($v_n \in L$ implies $-v_n \in L$), and has $0 \in L$.
\label{assum:level}
\end{assum}
\begin{assum}
    The vector $q$ has bounded entries. \label{assum:harmonic}
\end{assum}
These assumptions help define a well-defined OPP instance. 
Assumptions \ref{assum:switch} and \ref{assum:interlock} together ensure lack of Zeno execution (no infinite number of switches in a finite time) \cite{teel2012hybrid}. Assumption \ref{assum:level} is a property of the converter topology. Assumption \ref{assum:harmonic} ensures bounded harmonics constraints.


\subsection{Hybrid System Description}


Each mode of the hybrid system is indexed by $(n, i)$, using the modulation output level $n \in 1..N$ and the number of elapsed switching transitions $i \in 0..d$. 
A transition graph $\gs$ with vertices $\vs$ and edges $\es$ can be formed to represent the switching sequence. Denoting the central level as $N_c = (N+1)/2$, the set of vertices in $\gs$ are
\begin{align}
    \vs = \left\{(n, i) \mid \begin{array}{ll}
         n \in 1..N, \ \  i \in 1..d \\
         \abs{n-N_c} \leq i, \ \text{mod}(i +n, 2)=0
    \end{array} \right\}.
\end{align}
The total number of vertices in this graph is $\vs = (\lceil{d/2}\rceil + 1) + 2 \sum_{i=1}^{N_c} (\lceil{(d-i)/2}\rceil+1)$.
The step-up and step-down edges between nodes in $\vs$ are 
\begin{align}
    \es^\pm &= \{(n\mp1, i) \rightarrow (n, i+1)\}_{(n, i)}.
\end{align}
The total edge set is $\es = \es^+ \, \cup \,  \es^-$. The total number of edges in this transition graph is $\abs{\es} = 2d - (N_c - 1)N_c$.  Unipolar patterns only allow for nonnegative-valued modulation levels in the first quarter-period $(\forall \theta \in [0, \pi/2]: u(\theta) \geq 0)$. The unipolar transition graph $\gs^{\text{uni}}$ is a subset of $\gs$ keeping only vertices with $n \geq N_c$.
We denote $\mathcal{P}$ the set of paths (sequences of arcs) in the transition graph $\gs$ starting $(N_c, 0)$ and ending at $(n, d)$ for some $n \in 1..N$. 
We denote by $\mathcal{P}$ the set of paths (sequences of arcs) in the transition graph $\gs$ starting $(N_c, 0)$ and ending at $(n, d)$ for some $n \in 1..N$. 

Figure \ref{fig:transition_graph} visualizes the graphs $\gs$ (top) and $\gs^{\text{uni}}$ (bottom) for a signal with $d=8$ switching angles and $N=7$ levels. The black nodes are the vertices indexed by $(n, i)$. The blue arrows represent step-ups $\es^+$ and the red arrows detail step-downs $\es^-$. 

\begin{figure}[t!]
    \centering

     \begin{subfigure}[b]{0.65\linewidth}
         \centering
         \includegraphics[width=\textwidth]{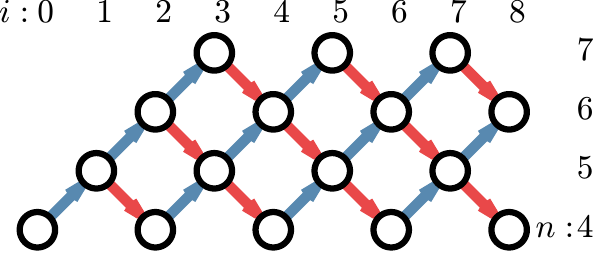}
         \caption{Multipolar Transitions $\gs$}
         \label{fig:y equals x}
     \end{subfigure}
     \vspace{0.5cm}
  
     \begin{subfigure}[b]{0.65\linewidth}
         \centering
         \includegraphics[width=\textwidth]{img/transition_qw_unipolar.pdf}
         \caption{Unipolar Transitions $\gs^{\text{uni}}$}
         \label{fig:three sin x}
     \end{subfigure}
    \caption{Transition graphs for $d=8, N=7$}
    \label{fig:transition_graph}
    \vspace{-6pt}
\end{figure}

\subsubsection{Mode and Switch Dynamics}

The states at each mode are the cosine-angle $c$, sine-angle $s$, clock angle $\phi$, and inductor current $I$. The phase angle $\theta$ is mapped as $(c, s) = \psi(\theta)$. The clock angle $\phi$ records the amount of time (angle) as the last transition, and is used to ensure that the interlocking angle constraint is satisfied. These variables are stacked into $x = [c, s, \phi, I]$. The dynamics at each mode $(n, i) \in \vs$ depend only on the converter level $n$ with 
\begin{align}
    \dot{c} &= -s &  \dot{s} &= c &  \dot{\phi} &= 1 & \dot{I} &= v_n.\label{eq:mode_dynamics}
\end{align}
The harmonics constraint in \eqref{eq:tdd_orig_alpha_she} can be cast as integral constraints over polynomials in $(c, s)$ using Chebyshev polynomials of the second kind $U_\ell$:
\begin{align}
    U_0(c) &= 1, \  U_1(c) = 2c, \ U_{\ell+1}(c) = 2c U_{\ell}(c) - U_{\ell-1}(c) \nonumber  \\
    \sin(\ell \theta) &= \sin(\theta) U_{\ell-1}(\cos(\theta)) \rightarrow s U_{\ell-1}(c).
\end{align}

The overall state space in which the variables $x$ can take values is $X = B \times [0, \pi/2] \times I_s$.
Due to the interlocking angle constraint (Assumption \ref{assum:interlock}), at least $\Theta i$ radians must elapse before entering mode $(n, i)$. Similarly, if the signal has a pulse number of  $d$, then mode $(n, i)$ must be exited at the latest by  $\pi/2 - \Theta(d-i)$ radians.  The maximum amount of time that can elapse before a switch is also $\pi/2 - \Theta d$, thus setting an upper-bound on the clock $\phi$.
Any \ac{QW} trajectory is required to start in the location $(n, i) = (N_c, 0)$, with an initial condition of
\begin{align}
    X^0_{N_c, 0} = \{1\} \times \{0\} \times  [\Theta/2, \pi/2-\Delta d] \times I_S.
\end{align}
If the trajectory $x(\theta)$ is operating in mode $(n, i) \in \vs$, its states can evolve within the set
\begin{align}
    \label{eq:support_mode}
    X_{n, i} &= B([\Theta i, \pi/2 - \Theta (d-i)]) \times [0, \pi/2 - \Theta d] \times I_s. 
\end{align}

Figure \ref{fig:support_circ} visualizes the support region for $(c, s)$ in \eqref{eq:support_mode} for a pulse number $d=4$ ($k=16$) and an interlocking angle of $\Theta = \pi/18$. The solid black dot in each pane is the initial point $(c, s) = (1, 0)$, and the unfilled dot is the terminal point $(c, s) = (0, 1)$.

\begin{figure}[t!]
    \centering
    \includegraphics[width=\linewidth]{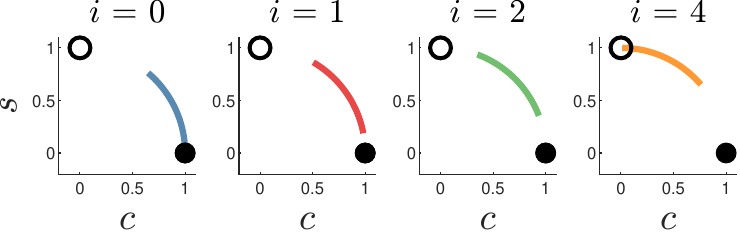}
    \caption{Valid regions for $(c, s)$ under $d=4, \ \Theta = \pi/18$}
    \label{fig:support_circ}
    \vspace{-6pt}
\end{figure}

Each transition $e \in \es$ has identical clock-zeroing resets
\begin{align}
    c_+ &= c &  s_+ &= s &  \phi_+ &= 0 & I_+ &= I.\label{eq:jump_dynamics}
\end{align}

The arc $e : (n \mp 1) \rightarrow (n, i+1)$ may be traversed if the state $x(\theta)$ is a member of the guard set
\begin{align}
    G_{n, i}^\pm = &B([\Theta i, \pi/2 - \Theta (d-i)])\nonumber\\
    &\qquad\times [\Theta, \pi/2 - \Theta (d +1/2)] \times I_s. \label{eq:support_guard}
\end{align}

The guard set constraint of $\phi \geq \Theta$ in \eqref{eq:support_guard} enforces the the interlocking angle constraint \eqref{eq:tdd_orig_alpha_lock}.

The terminal nodes of any path are vertices with $(n, d) \in \vs$. The terminal state sets associated with these nodes are
\begin{align}
    \label{eq:support_mode_terminal}
    X_{n, d}^{\pi/2} &=  \{0\} \times \{1\} \times  [\Theta/2, \pi/2 - \Theta (d + 1/2)] \times I_s. 
\end{align}
The number of such terminal nodes is 
\begin{align}
    N_{\pi/2} = \min [(N-1)/2 + \text{mod}(d+1, 2),  \ d+1].
\end{align}

\subsubsection{Optimal Control Problem}

The OPP task in Problem \ref{prob:tdd_orig} involves the selection of switching angles $\{\alpha^i\}$  and level indices $\{n^i\}$. The hybrid system optimal control task is defined with respect to angles $\{\alpha^i\}$ and a path $P \in \mathcal{P}$. 
A path $P$ can be uniquely constructed from each $(\{\alpha^i\}, \{n^i\})$ feasible for constraints \eqref{eq:tdd_orig_step}-\eqref{eq:tdd_orig_alpha_n_end}. At each $i \in 1..d$, the incoming edge $P^{i}$ in the path $P$ is 
\begin{align}
    P^i =  (n^{i-1}, i-1) \rightarrow (n^i, i). \label{eq:arc_edge}
\end{align}
Furthermore, we define $\text{src}(P)$ as the initial modulation level $n^0 = \text{src}(P^1)$. 
Given a jump pattern $\mathcal{T}= (\{\alpha^i\}, P)$, we define Loc$(\theta; \mathcal{T})$ as a trajectory-dependent function that returns the current mode $(n, i)$ of the dynamics at angle $\theta$
\begin{align}
i^*(\theta; \mathcal{T}) &= \max_{i \in 0..d+1} i  \quad \text{s.t.} \ \   \theta \geq \alpha^i\\
    \text{Loc}(\theta; \mathcal{T})&= (n^{i^*(\mathcal{T})}, i^*(\mathcal{T})).
\end{align}
The location function $\text{Loc}$ can be used to define an indicator function $\chi_{n, i}^\mathcal{T}(\theta)$, which takes the value $1$ when $\text{Loc}(\theta; \mathcal{T}) = (n, i)$, and $0$ otherwise.
We further define a linear operator $\Lambda_\mathcal{T}$ as a mapping that integrates \ac{QW}-symmetric functions $z(x, u)$ as
\begin{align}
    \Lambda_\mathcal{T}: z(x, u) \mapsto 4 \sum_{n, i} \int_{\theta=0}^{\pi/2} &\chi_{n, i}^\mathcal{T}(\theta)z(x(\theta), v_n) d \theta \label{eq:lintegral_operator}
\end{align}


The objectives and harmonics constraints in \eqref{eq:tdd_orig_alpha} can be equivalently expressed using the operator $\Lambda_\mathcal{T}$ with respect to the functions $z(x, u) \in \{I^2, s U_{\ell-1}(c)\}$ for odd $\ell$.

The mode-selected hybrid system OCP interpretation for Problem \ref{prob:tdd_orig} is as follows:
\begin{prob}
\label{prob:tdd_hy}
For a given converter with parameters in Table \ref{tab:opp_tdd_param}, choose an initial condition $(\phi_0, I_0)$ and sequence $\mathcal{T}$ satisfying  
\begin{subequations}
    \label{eq:tdd_hy}
\begin{align}
    J^*_{hy} = &\inf_{\phi_0, I_0, \mathcal{T}} \quad \Lambda_\mathcal{T}[I^2]  \label{eq:tdd_hy_obj}\\
 \text{s.t. } \quad   & \Lambda_\mathcal{T}[s U_{q-1}(c)] \in h, \quad\label{eq:tdd_hy_she} \\
    & x(\theta) \text{ follows  \eqref{eq:mode_dynamics} when Loc}(\theta;\mathcal{T}) =  (n, i) \label{eq:tdd_hy_follow} \\
    & x(\theta) \text{ resets as \eqref{eq:jump_dynamics} when } \theta \in \{\alpha^i\} \\    
    & \phi(\theta) \geq \Theta \text{ when switching (constraint \eqref{eq:support_guard})} \label{eq:tdd_hy_guard}\\
    & \phi(0) =  \phi_0, \  I(0) =  I_0,  \label{eq:tdd_hy_clock}\\    
    & \phi_0 \in [\Theta/2, \pi/2-\Theta (d+1/2)],  \ I_0 \in I_S \\
    & \phi(\pi/2) \geq \Theta/2 \label{eq:tdd_hy_endlock}\\
    & \mathcal{T} = (\{\alpha^i\}_{i=}^{d}, P), \ P \in \mathcal{P}\\
    &\forall i \in 1..d-1: \  \alpha^{i+1} \geq \alpha^{i} \label{eq:tdd_hy_lock}\\    
    & \alpha_1 \geq 0, \ \alpha^d \leq \pi/2, \  \text{Loc}(0; \mathcal{T}) = (\text{src}(P), 0).     \label{eq:tdd_hy_start}
\end{align}
\end{subequations}
\end{prob}

\begin{prop}
Problems \ref{prob:tdd_orig} and  \ref{prob:tdd_hy} have the same objective $(J^* = J^*_{hy})$.    
\end{prop}
\begin{proof}
    Feasible points $\{\alpha^i\}, \{n^i\}$ of Problem \ref{prob:tdd_orig} can be mapped into representations $\mathcal{T}$ in a one-to-one manner, as noted by the transformation in \eqref{eq:arc_edge}. The reverse process may generally fail, because constraint \eqref{eq:tdd_orig_alpha_lock} enforces separation by $\Theta$, whereas \eqref{eq:tdd_hy_lock} lacks this interlocking angle constraint. Instead, the $\Theta$ dwell time constraint is enforced by  \eqref{eq:tdd_hy_guard}, thus assuring equivalence.    
\end{proof}
\section{Convex Linear Measure Program}
\label{sec:opp_lp}

Both the finite-dimensional $(\alpha, n)$ program in Problem \ref{prob:tdd_orig} and the hybrid OCP in Problem \ref{prob:tdd_hy} are nonconvex optimization problems. 
This section uses measure-theoretic methods from \cite{zhao2017optimal, lewis1980relaxation} to compute lower-bounds on the minimum signal energy using convex programming.

\subsection{Measure Preliminaries}
We first briefly review concepts and notation associated with measures \cite{tao2011introduction}. The set of continuous functions defined over a set $S \subseteq \R^n$ is $C(S)$, and the
set of nonnegative Borel measures supported on $S$ is $\Mp{S}$. 
A duality pairing $\inp{\cdot}{\cdot}$ exists between any function $f \in C(S)$ and measure $\mu \in \Mp{S}$ by Lebesgue integration: $\inp{f}{\mu} = \int f \ d \mu =  \int_S f(x) d \mu(x)$. 
The mass of a measure $\mu \in \Mp{S}$ is $\inp{1}{\mu}$, and $\mu$ is a \textit{probability measure} if $\inp{1}{\mu} = 1$.  The Dirac Delta $\delta_{x = x'}$ is a probability measure supported only at the point $x'$, thus satisfying the relation $\forall f \in C(S): \inp{f}{\delta_{x=x'}} = f(x')$. For a multi-index $\gamma \in \N^n$, the $\gamma$-moment of the measure $\mu$
 is $\inp{x^\gamma}{\mu} = \inp{\prod_{i=1}^n x^{\gamma_i}}{\mu}.$
\subsection{Involved Measures}


\begin{table}[t!]
    \centering
    \caption{Measure variables used in OPP}
    \begin{tabular}{l r  l l}
        Initial & $\mu^0$ & \hspace{-0.3cm}$\in \Mp{X^0}$ &   \\
        Terminal & $\mu^\partial_{n}$ &  \hspace{-0.3cm}$\in \Mp{X^{\pi/2}_{n, d}}$ & $n+d \in 2\Z$ \\
        Occupation & $\mu_{n, i} $ &  \hspace{-0.3cm}$\in \Mp{X_{n, i}}$ &  $(n, i) \in \vs$ \\
         Step Up & $\rho_{n, i}^+$ &  \hspace{-0.3cm}$\in \Mp{G_{n, i}^+}$ & \hspace{-0.2cm} $(n-1, i-1) \rightarrow (n, i)\in \es^+$ \\
         Step Down & $\rho_{n, i}^-$ &  \hspace{-0.3cm}$\in\Mp{G_{n, i}^-}$& \hspace{-0.2cm} $(n+1, i-1) \rightarrow (n, i)\in \es^-$ \\
    \end{tabular}
    \label{tab:meas_tdd}
    \vspace{-8pt}
\end{table}

Table \ref{tab:meas_tdd} lists measure variables that are involved in the transference from the OCP in Problem~\ref{prob:tdd_hy} to a convex LP.

The initial measure $\mu^0$ encodes initial conditions at $i=0$ before any switching. The terminal measure $\mu^\partial$ stores a distribution of  stopping points of the trajectory at mode $i=d$. The occupation measure $\mu$ tracks the evolution of the trajectory within each mode $(n, i)$. The mass of the occupation measure $\inp{1}{\mu_{n, i}}$ can be interpreted as the (averaged) amount of angular arc the switching signal spends in mode $(n, i)$. The jump measures $\rho^\pm$ keep information about the switching angle and active arc of the transition in the graph $\gs$. The support constraints in $\phi$ for  $\mu^\partial$ and $\rho^\pm$ encode the clock-based constraints \eqref{eq:tdd_hy_guard} and \eqref{eq:tdd_hy_endlock}, respectively.

\subsection{Construction Procedure}
\label{sec:construction}
We first review how measures in Table \ref{tab:meas_tdd} can be created from an initial current $I_0 \in I_S$ and a switching sequence $\mathcal{T} = (\alpha, P)$ feasible for constraints \eqref{eq:tdd_hy_follow}-\eqref{eq:tdd_hy_start}. Let $I(\theta)$ denote the inductor current when the switching sequence $u(\theta)$ is applied. 
The initial value of the clock at angle $\theta=0$ over this switching sequence is $\phi_0 = \pi/2 - \alpha^{d}$.
The initial measure is chosen as $\tilde{\mu}^0 = \delta_{\phi = \pi/2-\alpha^{d}, I = I_0}$. The stopping measures are picked as $\tilde{\mu}^\partial_n = \mu^0_n$. The switching measures are $\tilde{\rho}^\pm_{n, i} = \delta_{(c, s) = \psi(\alpha^i), \phi = \alpha^i - \alpha^{i-1}, I = I(\alpha^i)}$ if $n = \text{dst}(P^i)$ and $P^i \in \mathcal{E}^\pm$, and are zero otherwise. The occupation measure is chosen to satisfy a relation for all test functions $\forall w \in C(X)$:
\begin{align}
   \textstyle \inp{w}{\tilde{\mu}} = \int w(x) d \tilde{\mu}_{n, i}(x) = \int_{\theta=0}^{\pi/2} \chi_{n, i}^\mathcal{T}(\theta) w(x(\theta)) d \theta. \label{eq:occ_int}
\end{align}
Under the definition in \eqref{eq:occ_int}, the current signal energy satisfies $\Lambda_\mathcal{T}[I^2]= 4\sum_{(n, i)} \inp{I^2}{\tilde{\mu}_{n, i}}$,
and the specified harmonics related for odd $\ell$ by
\begin{align}
    \sin(\ell \theta): \ \ \Lambda_\mathcal{T}[s U_{\ell-1}(c)] = 4 \textstyle \sum_{(n, i)} \inp{s U_{\ell-1}(c)}{\mu_{n, i}}. \nonumber
\end{align}

The signal energy and harmonics constraints are thus preserved when transiting from the $(\alpha, \mathcal{T})$ switching sequence representation to the measure embedding in Table \ref{tab:meas_tdd}.

\subsection{Convex Reformulation}

We now describe how properties of the hybrid  OCP in Problem \ref{prob:tdd_hy} can be encoded as convex expressions in the measures from Table \ref{tab:meas_tdd} following \cite{zhao2017optimal, lewis1980relaxation}.

\subsubsection{Objective}
The objective is reformulated into $\norm{I}_2^2  \rightarrow \Lambda_{\mathcal{T}}[\norm{I}_2^2] \rightarrow \textstyle 4 \sum_{n, i}  \inp{I^2}{\mu_{n, i}(x)}$, which is linear in $\mu$.

\begin{rmk}
    If the target to minimize is the signal energy of the voltage  $u(\theta)$ (for voltage TDD), then the objective can instead be chosen as $\norm{u}^2_2 \rightarrow 4 \sum_{n, i} v_n^2 \inp{1}{\mu_{n, i}(x)}$.
\end{rmk}

\subsubsection{Initial}
The initial distribution should be a probability measure, which includes as a special case a measure solution constructed from a single trajectory (Section \ref{sec:construction}).
Due to the \ac{QW} symmetry, only the central level may be a possible initial location at $i=0$. As such, the initial constraint is 
\begin{align}
 \inp{1}{\mu^0_{(N+1)/2}} = 1.\label{eq:con_init}
\end{align}

\subsubsection{Harmonics}
The harmonics constraint \eqref{eq:tdd_hy_she} is respected across the occupation measures as
\begin{align}
     &\textstyle 4 \sum_{n, i} v_n \inp{s U_{q-1}(c)}{\mu_{n, i}(x)} \in h. \label{eq:con_harmonic} 
\end{align}

\subsubsection{Uniformity} The signal $u(\theta)$ must be entirely defined between the angles $\theta \in [0, \pi/2]$. This inspires a constraint on the occupation measures $\forall \eta \in C(B([0, \pi/2])):$
\begin{align}
      \textstyle \sum_{n, i} \inp{\eta}{\mu_{n, i}} = \textstyle \int_{\theta=0}^{\pi/2} \eta(\psi(\theta)) d \theta. \label{eq:con_uniformity}
\end{align}

\subsubsection{Conservation}

The measures in Table \ref{tab:meas_tdd} are connected together by the mode dynamics \eqref{eq:mode_dynamics} and jump dynamics \eqref{eq:jump_dynamics}. This connection is enforced with a conservation law (Liouville/Martingale equation) \cite{lewis1980relaxation}.
Figure \ref{fig:continuity} visualizes the continuity relationship between measures, which is used to develop Liouville-type flow constraints. Trajectories enter mode $(n, i)$ through incoming jumps $(\rho^\pm_{n \mp i})$ and through the initial condition if $i=0$ $(\mu^0)$. Trajectories depart mode $(n, i)$ upon stopping at the $d$-th swtitching transition $(\mu^\partial_{n})$,  or switching out of the mode ($\rho^\pm_{n, i+1}$).

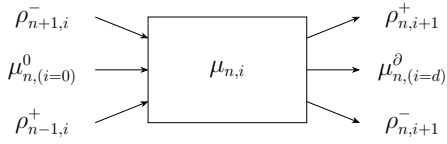
\begin{figure}[t!]
\centering
\resizebox{0.7\linewidth}{!}{%
\begin{circuitikz}
\tikzstyle{every node}=[font=\LARGE]
\draw  (7.5,13) rectangle  node {\LARGE $\mu_{n, i}$} (11.25,10.5);
\draw [->, >=Stealth] (6.25,11.75) -- (7.5,11.75);
\draw [->, >=Stealth] (11.25,11.75) -- (12.5,11.75);
\draw [->, >=Stealth] (11.25,12.5) -- (12.5,13);
\draw [->, >=Stealth] (11.25,11) -- (12.5,10.5);
\draw [->, >=Stealth] (6.25,10.5) -- (7.5,11);
\draw [->, >=Stealth] (6.25,13) -- (7.5,12.5);
\node [font=\LARGE] at (5,11.75) {$\mu_{n, (i=0)}^0$};
\node [font=\LARGE] at (5,10.5) {$\rho_{n-1, i}^+$};
\node [font=\LARGE] at (5,13) {$\rho_{n+1, i}^-$};

\node [font=\LARGE] at (13.75,11.75) {$\mu_{n, (i=d)}^\partial$};
\node [font=\LARGE] at (13.75,10.5) {$\rho_{n, i+1}^-$};
\node [font=\LARGE] at (13.75,13) {$\rho_{n, i+1}^+$};
\end{circuitikz}
}%
\caption{Continuity relation between measures}
\label{fig:continuity}
    \vspace{-6pt}
\end{figure}

The Lie derivative $\Lie_n$ of a function $w \in C^1(X)$ along dynamics in \eqref{eq:mode_dynamics} at mode $(n, i)$ is
\begin{align}
    \Lie_{n, i} w(x) = (-s \partial_c + c \partial_s + \partial_\phi + v_n \partial_I)w(x).
\end{align}
Given a function $w \in C^1(X)$ and an angle $\alpha$, we define the following notational abbreviations:
\begin{align}
    w_\alpha(\phi, I) &= w(\cos(\alpha), \sin(\alpha), \phi, I)  \\
    w_\phi(c, s, I) &= w(c, s, 0, I).    
\end{align}
The substitution $w_\alpha$ is used at the initial angle $\alpha^0 = 0$ and the terminal angle $\alpha^{d+1} = \pi/2$. The symbol $w_\phi$ implements the reset dynamics \eqref{eq:jump_dynamics}.
The conservation relation holding for all levels $n \in 1..N$ and functions $w \in C(X)$ (using the convention $\mu_{n, i}=0$ if $(n, i) \not\in \vs$ with similar formulations for $\rho^\pm, \mu^{\partial}$) is:
\begin{align}
    &\inp{w_0}{\mu^0} + \inp{\Lie_{N_c, 0}w}{\mu_{N_c, 0}} = \inp{w}{\rho^+_{N_c+1, 1} + \rho^{-}_{N_c-1, 1}} \nonumber\\   
    & \forall i \in 1..d-1:\label{eq:con_conservation}  \\
    & \  \ \inp{w_\phi}{\rho^-_{n+1, d} + \rho^+_{n-1, d}} +  \inp{\Lie_{n, i} w}{\mu_{n, i}} = \inp{w}{\rho^+_{n, 1} + \rho^{-}_{n, 1}} \nonumber\\
    &\inp{w_\phi}{\rho^-_{n+1, d} + \rho^+_{n-1, d}} + \inp{\Lie_{n, d} w}{\mu_{n, d}} = \inp{w_{\pi/2}}.{\mu^\partial_{n}}\nonumber
\end{align}
This continuity relation ensures that the (expected) value of $w$ at the the starting point $\theta=0$ equals the (expected) value of $w$ at the end point $\theta=\pi/2$ minus the accumulated change in $w$ over all flows and transitions \cite{zhao2017optimal}.

    
\subsection{Measure Program Formulation}

\begin{prob}
\label{prob:tdd_lp}
    Given parameters in Table \ref{tab:meas_tdd}, solve
    \begin{subequations}        
    \label{eq:tdd_lp}
    \begin{align}
        p^* = &\inf \textstyle\ 4 \sum_{n, i}  \inp{I^2}{\mu_{n, i}(x)} \\
        & \text{Constraints \eqref{eq:con_init}, \eqref{eq:con_harmonic}, \eqref{eq:con_uniformity}, \eqref{eq:con_conservation}} \label{eq:con_all}\\
        & \text{Measures from Table \ref{tab:meas_tdd}.} 
    \end{align}
    \end{subequations}
\end{prob}

\begin{thm}
    The LP in \eqref{eq:tdd_lp} lower-bounds the original OPP problem in \eqref{eq:tdd_orig_alpha} with $p^* \leq P^*$.
\end{thm}
\begin{proof}
    The construction procedure from Section \ref{sec:construction} transforms any switching sequence into a set of measures from Table \ref{tab:meas_tdd}. Feasibility of the constraints of \eqref{eq:tdd_orig_alpha} implies that the measure solution is feasible for the constraints in \eqref{eq:con_all}. Because the objective $\norm{I}^2_2$ is the same between the switching and measure representations, the lower bound $p^* \leq P^*$ is therefore established.
\end{proof}

\begin{rmk}
    We conjecture that $p^*=P^*$ under Assumptions \ref{assum:switch}-\ref{assum:harmonic}, but this exactness has yet to be established (due to both controlled switching and integral constraints).
\end{rmk}
\subsection{Moment Truncation}

The LP in Problem \ref{prob:tdd_lp} involves infinite-dimensional measure variables from Table \ref{tab:meas_tdd}. The moment-SOS hierarchy of semidefinite programs (SDPs) offers one method to truncate the infinite-dimensional LP into  finite-dimensional convex optimization problems \cite{lasserre2009moments}. These problems are indexed by a polynomial degree $\beta \in \N$.
The variables of the finite-dimensional convex programs are vectors $\mathbf{y} = (\mathbf{y}^0, \mathbf{y}^\partial, \mathbf{y}^{\text{occ}}, \mathbf{y}^\pm)$ appropriately indexed by $(n, i)$, corresponding to vectors of possible moments for the relevant measures. These vectors must satisfy linear matrix inequality (LMI) constraints in order to plausibly be moments of measures in Table \ref{tab:meas_tdd}.
Table \ref{prob:tdd_lp} reviews information about the finite-degree truncation of Problem \ref{prob:tdd_lp}, including the dimension of the maximal-size positive semidefinite (PSD) matrix in the moment-constraining LMI for the vector.

\begin{table}[t!] 
\caption{Properties of the degree-$\beta$ \ac{OPP} truncation}
\vspace{-10pt}
\begin{align}
\def\arraystretch{1.5}
    \begin{array}{rcccc}          
         \text{Variable:} & \mathbf{y}^0 & \mathbf{y}^\partial_n & \mathbf{y}_{n, i}^{\text{occ}} & \mathbf{y}^\pm_{n, i}  \\
         \text{Relevant Measure:} & \mu^0 & \mu^\partial_n & \mu_{n, i} & \rho^\pm_{n, i}  \\
         \text{Multiplicity:} & 1 & N_{\pi/2} & \abs{\vs} & \abs{\es} \\
         \text{Involved States:} & \phi, I & \phi, I & x & x \\
         \text{Vector Dim.:} & \binom{2+2\beta}{2} & \binom{2+2\beta}{2} & \binom{4+2\beta}{4} & \binom{4+2\beta}{4} \\
         \text{Max. PSD Dim:} & \binom{2+\beta}{2} & \binom{2+\beta}{2} & \binom{4+\beta}{4} & \binom{4+\beta}{4}
    \end{array} \nonumber
\end{align}
    \label{tab:mom_tdd}
    \vspace{-8pt}
\end{table}

Let $p_\beta$ denote the optimal value of the degree-$\beta$-truncation of Program \eqref{prob:tdd_hy} (SDP in which the vectors $\mathbf{y}$ correspond to plausible moments up to degree $2\beta$).
Increasing the polynomial degree $\beta$ enforces additional constraints on the sequences in  $\mathbf{y}$, and produces a rising sequence of lower-bounds $p^*_{\beta} \leq p^*_{\beta+1} \leq \ldots \leq p^*$. The per-iteration complexity of solving for the bound $p^*_\beta$ scales as $O((\abs{\vs} + \abs{\es})^{1.5} \beta^{22.5})$ \cite[Section 3.4]{claeys2016modal}. 
Given a feasible solution $\mathbf{y}$ at degree $\beta$, for every $(n, i) \in \vs$ let  $\xi_{n, i}$ denote the first entry of each $\mathbf{y}^{\text{occ}}_{n, i}$. The value $\xi_{n, i}$ corresponds to the mass $\inp{1}{\mu_{n, i}}$ (averaged time spent in mode $(n, i)$). A switching sequence $\mathcal{T}^{\text{rec}}$ can be recovered from the occupancy values $\xi$ as
\begin{subequations}
\label{eq:sequence_rec}
\begin{align}
    \forall i \in 0..d: \quad  &  u^i = \argmax_{n \in 1..N} \xi_{n, i}, \  (n, i) \in \vs \label{eq:sequence_rec_u}\\
    \forall i \in 1..d: \quad& \textstyle\alpha^i = \sum_{i' \in 0.. i-1} \ \sum_{n \mid (n, i) \in \vs} \xi_{n, i}.
\end{align}
\end{subequations}
While the recovered sequence $\mathcal{T}^{\text{rec}}$ may violate the harmonics specifications, $\mathcal{T}^{\text{rec}}$  may be employed as an initial point for local search algorithms. A possible approach is to use \texttt{fmincon} to optimize Problem \ref{prob:tdd_orig} over $\alpha$ under a fixed sequence of transitions $u$ from \eqref{eq:sequence_rec_u}.
\section{Numerical Examples}

\label{sec:examples}

The MATLAB (R2024a) code to generate all examples is publicly available.\footnote{\url{https://github.com/jarmill/opp_pop}} Dependencies for this code include GloptiPoly 3 \cite{henrion2009gloptipoly} to generate the moment relaxations, YALMIP \cite{lofberg2004yalmip} to parse the moment programs, and Mosek \cite{mosek110} to solve the resultant LMIs via primal-dual interior-point semidefinite programming.
The examples in this section involve control of a five-level converter with levels
    $L = \begin{bmatrix}
        -1, \ -0.5, \ 0, \ 0.5,  \ 1
    \end{bmatrix}.$
All examples consider \ac{QW} symmetry and unipolar transitions, and include the harmonic inequality constraint $b_3 \in [-0.01, 0.01]$. The modulation index $M$ is used as a problem parameter with equality constraint $b_1 = M$. The fundamental frequency of $f_1 = 50\,$Hz and interlocking time of $T_s = 10^{-4}\,$s induce an interlocking angle of $\Theta = \pi/100\,$rads. The solutions are reported in terms of the current TDD when ignoring constant factors in \eqref{eq:tdd_current}, yielding the quality metric 
    $Q = \sqrt{\norm{I}_2^2/\pi - b_1^2}$.
Note that $Q_\beta$ refers to a lower bound on the minimal $Q$ computed from a degree-$\beta$ truncation of Problem \ref{prob:tdd_lp}. A switching sequence $\mathcal{T}$ is considered to be (numerically) feasible if $b_1 \in [M, M + 10^{-7}]$ and $b_3 \in [-0.01, 0.01]$.

\subsection{Single Pattern}

This first example focuses on synthesis of an \ac{OPP} with  $M=0.9$ and $d=8$. Figure \ref{fig:q32} reports SDP lower bounds on the minimal $Q$ as obtained by solving the degree-$\beta$ truncation of Problem \ref{prob:tdd_lp} between $\beta \in 1..7$. The black line on the top pane is the $Q$ value from an \texttt{fmincon}-recovered feasible OPP with $Q^{rec} = 1.16004\times 10^{-2}$. This recovered OPP is described by
\begin{align}
      \{\alpha^i\}_{i=1}^8 &= \begin{Bmatrix}
          0.2020 &     0.2842 &     0.3645 &     0.8636 \\
    0.9900 &     1.1153 &     1.3343 &     1.4172
      \end{Bmatrix} \label{eq:recovered_opp}\\
      \{u^i\}_{i=0}^8 &= \begin{Bmatrix}
          0 & 0.5 & 0 & 0.5 & 1 
          & 0.5 & 1 & 0.5 & 1
      \end{Bmatrix}, \nonumber
\end{align}
and it has Fourier coefficients of $b_1 = 0.9 + 3.0758\times10^{-8}$ and $b_3 = -3.3773 \times 10^{-3}$.
The certified optimality gap for the pattern in \eqref{eq:recovered_opp} is $Q^{rec} - Q_6 = 1.33\times 10^{-5}$. The lower panel of Fig.~\ref{fig:q32} reports the times required to compile the LMI constraints (preprocess) and run Mosek (solver). 

\begin{figure}[t!]
    \centering
    \includegraphics[width=0.8\linewidth]{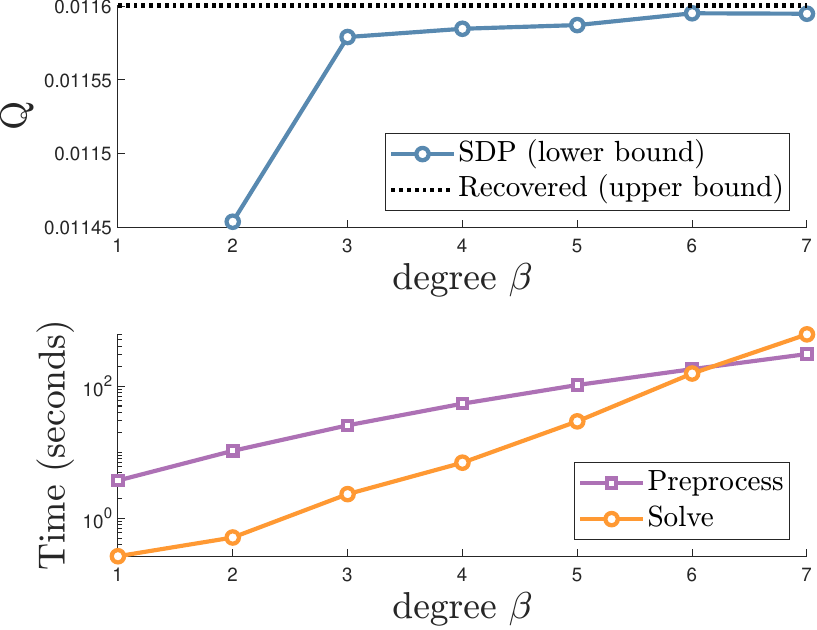}
    \caption{$Q$ bounds and timing for $d=8$, $M=0.9$}
    \label{fig:q32}
\end{figure}

Figure \ref{fig:N5_k32} plots the \ac{OPP} $u(\theta)$ from \eqref{eq:recovered_opp} in blue on the top subplot, and the inductor current $I(\theta)$ in orange in the middle plot. The black curves are the reference voltage $u^*(\theta) = 0.9 \sin(\theta)$ on top and the reference current $I^*(\theta) = -0.9 \cos(\theta)$ in the middle respectively. The bottom pane plots the point-wise residual $I(\theta) - I^*(\theta)$ in green.

\begin{figure}[t!]
    \centering
    \includegraphics[width=0.8\linewidth]{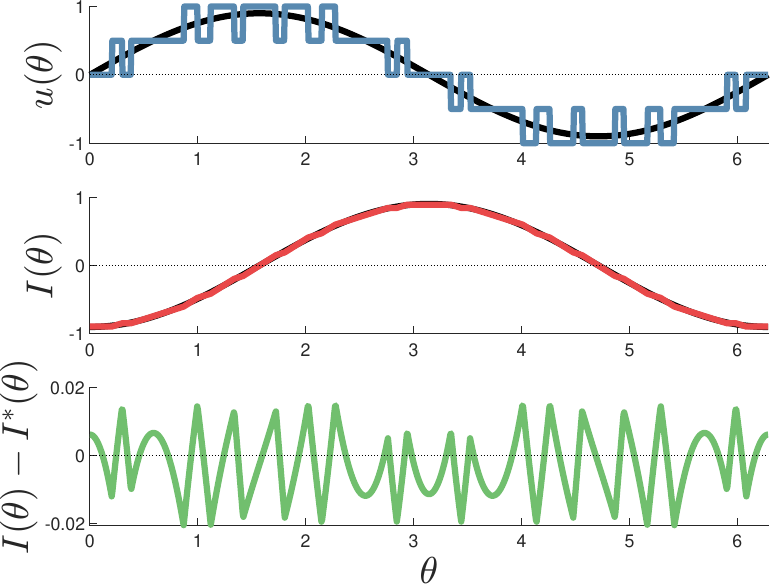}
    \caption{Computed OPP for $k=32$, $M=0.9$}
    \label{fig:N5_k32}
    \vspace{-6pt}
\end{figure}

\subsection{Parameter Sweep}
The second experiment computes the SDP lower bound $Q_3$ over a parameter sweep over $d \in 1..10$ and $M \in 0.05(1..22)$. 
Figure \ref{fig:N5_bound} plots the computed lower bounds $Q_3$ as the pulse number $d$ increases (x-axis) and as the modulation index increases (color).  The plot also shows the infeasibility of fulfilling the modulation requirement $M \geq 0.6$ when $d=1$.
Figure \ref{fig:sweep_gap} plots the logarithm of the difference between the current TDD of the recovered pulse pulse pattern $(Q_{rec})$ and the SDP lower bound reported at the degree-3 truncation of Program \ref{prob:tdd_lp} $(Q_3)$. The white areas are regions where the SDP lower-bound is infeasible, or (if feasible) where \texttt{fmincon} fails to recover a feasible switching sequence.

\begin{figure}[t!]
    \centering
    \includegraphics[width=0.85\linewidth]{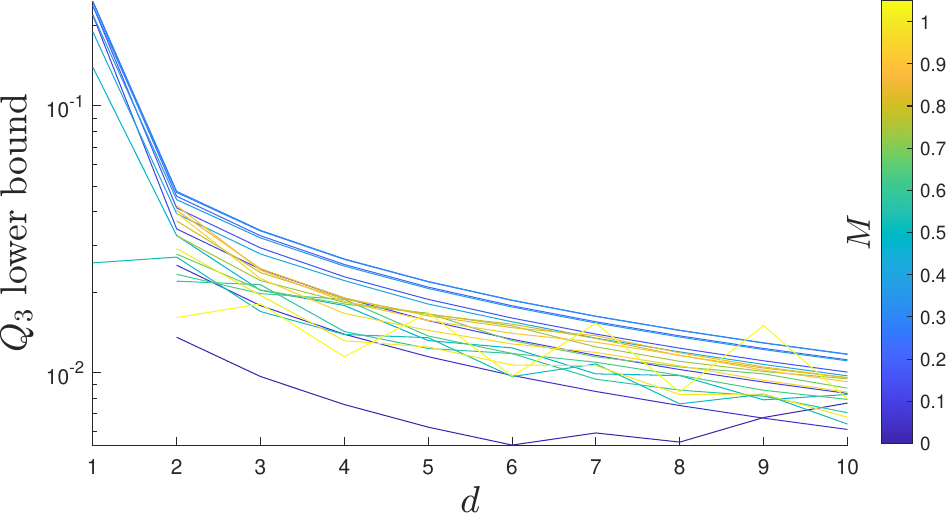}
    \caption{$Q_3$ Lower-bounds v.s.~$M$ and $d$ }
    \label{fig:N5_bound}
\end{figure}

\begin{figure}[t!]
    \centering
    \includegraphics[width=0.9\linewidth]{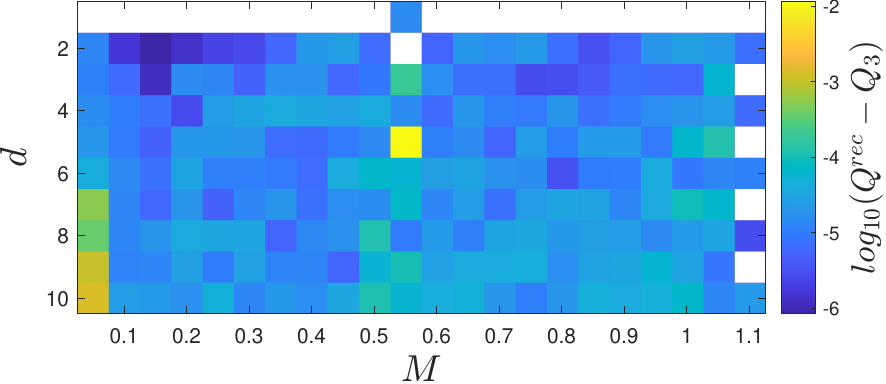}
    \caption{Optimality gaps for $Q_3$-recovered \acp{OPP}}
    \label{fig:sweep_gap}
    \vspace{-6pt}
\end{figure}

\section{Conclusion}

\label{sec:conclusion}

OPPs are an advanced modulation technique that can significantly improve the performance of power electronic systems, particularly at low pulse numbers. However, their computation requires solving highly nonconvex (mixed-integer) optimization problems. This work demonstrated how the \ac{OPP} optimization problem can be embedded into an OCP for a hybrid system, which was in turn lifted into a convex LP in measures. 
The derived LP was then lower-bounded using the moment-SOS hierarchy, yielding tractable bounds on the minimum achievable output current TDD under the imposed constraints. Future work includes formulating conditions for which $\lim_{\beta \rightarrow \infty} p_\beta =P^*$ and performing experimental validation of the computed OPPs in a real-world setting.




\appendices



\bibliographystyle{IEEEtran}
\bibliography{references.bib}

\end{document}